\documentclass[12pt,english]{article}
\usepackage[authoryear]{natbib}
\usepackage{lmodern}
\usepackage[LGR,T1]{fontenc}
\usepackage[latin9]{inputenc}
\usepackage{geometry}
\usepackage[active]{srcltx}
\usepackage{color}
\usepackage{babel}
\usepackage{amsmath,amssymb,amsthm}
\usepackage{graphicx}
\usepackage{setspace}
\usepackage[unicode=true,pdfusetitle,
 bookmarks=true,bookmarksnumbered=false,bookmarksopen=false,
 breaklinks=false,pdfborder={0 0 1},backref=false,colorlinks=true]
 {hyperref}
\hypersetup{
 pdfborderstyle=,citecolor=blue,linkcolor=blue,urlcolor=blue}

\makeatletter

\providecommand{\tabularnewline}{\\}

\theoremstyle{remark}

\theoremstyle{plain}
\newtheorem{thm}{\protect\theoremname}
\theoremstyle{plain}
\newtheorem{lem}{\protect\lemmaname}
\theoremstyle{remark}
\newtheorem{claim}{\protect\claimname}
\theoremstyle{plain}
\newtheorem{proposition}{\protect\propositionname}

\definecolor{green}{RGB}{0, 180, 0}
\definecolor{brown}{RGB}{140, 40, 40}
\providecommand{\claimname}{Claim}
\providecommand{\lemmaname}{Lemma}
\providecommand{\remarkname}{Remark}
\providecommand{\theoremname}{Theorem}

\providecommand{\lemmaname}{Lemma}
\providecommand{\propositionname}{Proposition}

\makeatother

\begin{document}
\title{Evolution, Heritable Risk, and Skewness Loving\thanks{We thank Erol Akcay, Gilad Bavly, Ben Golub, Aviad Heifetz, Laurent
Lehmann, John McNamara, Jonathan Newton, Debraj Ray, Roberto Robatto,
Larry Samuelson, Bal$\grave{\textrm{a}}$zs Szentes, the editor, two anonymous referees, and LEG2019 conference participants
for helpful comments. We thank Ron Peretz for his kind help in extending
Theorem 1. We thank Renana Heller for writing the simulation used
in Section \ref{sec:Numeric-Analysis-of}.}}
\author{Yuval Heller\thanks{Bar-Ilan University, Department of Economics. (email: \protect\protect\protect\protect\protect\href{mailto:yuval.heller\%5C\%5C\%5C\%5C\%5C\%5C\%5C\%5C\%5C\%5C\%5C\%5C\%5C\%5C\%5C\%5C\%5C\%5C\%5C\%5C\%5C\%5C\%5C\%5C\%5C\%5C\%5C\%5C\%5C\%5C\%5C\%40biu.ac.il}{yuval.heller@biu.ac.il},
homepage: \protect\protect\protect\protect\protect\href{https://sites.google.com/site/yuval26/}{https://sites.google.com/site/yuval26/}).
Heller is grateful to the European Research Council for its financial
support (ERC starting grant \#677057).}~ and Arthur Robson\thanks{Simon Fraser University, Department of Economics (email: \protect\protect\protect\protect\protect\protect\protect\protect\protect\protect\href{mailto:\%5C\%5C\%5C\%5C\%5C\%5C\%5C\%5C\%5C\%5C\%5C\%5C\%5C\%5C\%5C\%5C\%5C\%5C\%5C\%5C\%5C\%5C\%5C\%5C\%5C\%5C\%5C\%5C\%5C\%5C\%5C\%20robson\%5C\%5C\%5C\%5C\%5C\%5C\%5C\%5C\%5C\%5C\%5C\%5C\%5C\%5C\%5C\%5C\%5C\%5C\%5C\%5C\%5C\%5C\%5C\%5C\%5C\%5C\%5C\%5C\%5C\%5C\%5C\%40sfu.ca}{ robson@sfu.ca},
homepage: \protect\protect\protect\protect\protect\protect\protect\protect\protect\protect\href{https://sites.google.com/view/arthurrobson/home/}{https://sites.google.com/view/arthurrobson/home/}).
Robson thanks the {Social Sciences and Humanities Research Council
of Canada} for support. }}
\maketitle
\begin{abstract}
Our understanding of risk preferences can be sharpened by considering
their evolutionary basis. The existing literature has focused on two
sources of risk: idiosyncratic risk and aggregate risk. We introduce
a new source of risk, heritable risk, in which there is a positive
correlation between the fitness of a newborn agent and the
fitness of her parent. Heritable risk was plausibly common in
our evolutionary past and it leads to a strictly higher growth rate
than the other sources of risk. We show that the presence
of heritable risk in the evolutionary past may explain the tendency
of people to exhibit skewness loving today.

of preferences, risk attitude, risk interdependence, long-run growth
rate, fertility rate.\\
Final pre-print of a manuscript accepted for publication in \textit{Theoretical Economics}.

\end{abstract}

\section{Introduction \label{sec:Introduction}}

Our understanding of risk preferences can be sharpened by considering
their evolutionary basis (see \citealp{robson2011evolutionary}, for a
survey).
This claim was advanced in the economics literature by \cite{robson1996biological}, for example, who presented a model in which each agent lives a single period and faces a choice between lotteries
over the number of offspring. (See also related models in \citealp{lewontin1969population,mcnamara1995implicit}.)
Some of the feasible lotteries involve aggregate risk
(when all agents obtain the same realization). \cite{robson1996biological}
showed that idiosyncratic risk (independent
across individuals) induces a higher long-run growth
rate {(henceforth ``growth rate'')} than aggregate risk, and as a result natural selection should
induce agents to be more risk averse with respect to aggregate risk.\footnote{See \cite{heller2014overconfidence} for a discussion of why this
might explain people's tendency to overestimate the accuracy of their
private information.}

This result has been put into an intriguing new light by \cite{robatto2017biological}
who reconsider the model in continuous time. In such a framework it
is appealing to formulate both consumption and the production of offspring
as \emph{rates}. Once this is done aggregate risk becomes equivalent
to idiosyncratic risk as long as fertility and mortality are age-independent. (See \citealp{Robson-Samuelson-2019}, and Section \ref{sec:Discussion} of this paper.)

The way in which idiosyncratic risk has been modeled in the previous
literature captures well coin flips concerning fertility that only
affect a particular individual. However, it is compelling that, in
the evolutionary past, there were plausibly many cases in which
the ``outcome of the flip'' persisted from parents to offspring.
In this paper we capture this persistence by introducing a new source
of risk, heritable risk, which is basically idiosyncratic risk, but
allows a positive correlation between the fitness of a newborn agent
and the fitness of her parent.

Heritable risk in this sense must have been common in the evolutionary past of human beings.
Such risk is induced if the agent's fitness is heritable due to imitation of the parent's behavior or
genetic inheritance. For example, a foraging
technique in prehistoric hunter-gatherer societies would be inherited if an individual copied her parent's technique. Alternatively, risk is heritable if the choice an individual makes is controlled genetically, and this gene is passed down from mother to daughter. The key properties are just: (1) there is a positive correlation between the fitness of an agent
and that of her parent, and (2) by contrast, there is little
correlation between the fitness of two randomly chosen agents in the
population.

We show that this heritable risk yields a strictly higher growth rate
than the other sources of risk. We derive this result in \citeauthor{robatto2017biological}'s (\citeyear{robatto2017biological}) setup, as it is more striking to see the advantage of heritable
risk in a setup in which all other sources of risk are equivalent. It
is relatively simple to show that heritable risk is also advantageous
in other setups considered in the literature.

\paragraph{Highlights of the model}

Consider a simple setup in which agents occasionally redraw a lottery
over their consumption rate, and the realized consumption determines
the fertility rate through a concave increasing function $\psi$.
Specifically, assume that the lottery can yield a high consumption
rate ($c_{h}$, inducing a fertility rate $r_{h}=\psi\left(c_{h}\right)$)
with probability $q_{h}$ or a low consumption rate ($c_{l}$, inducing a fertility rate
$r_{l}=\psi\left(c_{l}\right)$) with probability $q_{l}=1-q_{h}$.
Each agent redraws her realized level of fertility at an annual rate
of $\lambda$. For simplicity assume that there is no mortality. Our
crucial departure from the existing literature is to assume that a
newborn agent inherits the realized fertility rate of her parent and
the values remain the same until either the parent or the offspring
redraws their fertility rate.

\paragraph{Key result}
{Theorem \ref{thm:main-result} shows that the growth rate $x^{*}$ 
induced by heritable risk 
is both (1) strictly higher than the lottery's expected
fertility rate
$\mu\equiv q_{\ell}\cdot r_{\ell}+q_{h}\cdot r_{h}$, but $x^{*}\rightarrow\mu$ as 
$\lambda\rightarrow \infty$,
and (2) strictly below the highest realization $r_{h}$, but $x^{*}\rightarrow r_h$ as $\lambda\rightarrow0$. 
To see the intuition behind (2), 
consider the case where $\lambda>0$ is small.}
The effect of the high realization of the heritable fertility rate gets
compounded over time since parents with high fertility rates beget offspring
with high fertility rates. Agents with high fertility rates therefore
form an increasing fraction of the population over time, causing the
overall growth rate to increase, and in the long run to  
{be close to}
$r_{h}$.

Our result has two main implications: (1) heritable risk induces a
higher growth rate than either aggregate risk or idiosyncratic risk
(both of which induce a growth rate that is equal to the lottery's
expectation $\mu$), and (2) this difference in the growth rates is
especially large when dealing with positively skewed lotteries (since the growth rate can be made close to $r_{h}$ in a way that is independent of the probability
$q_{h}$).

One can interpret our result as follows. The long-run impact of risk
interdependence depends on the ``direction'' of the interdependence
(vertical or horizontal). The form of risk we introduce induces correlation
between an agent's outcome and her offspring's outcome. {This ``vertical
correlation'' is helpful to the growth rate, as it allows successful
families to have fast exponential growth. 
By contrast, this risk does
not involve ``horizontal correlation'' of risk between agents of the same
cohort, which would be harmful to the growth rate.}
The insight that vertical correlation increases the growth rate, but horizontal correlation decreases it, may be applicable in other domains of economics and finance.

\paragraph{Risk attitude}

We assume that individuals in our evolutionary past had different
types, and that the agent's type determines her risk attitude---in particular, how the
agent chooses between a risky consumption option and a safe one. An agent is likely to have the same
type as her parent due to genetic inheritance. Occasionally, new types
are introduced into the population following a genetic mutation.
Observe that the population share of agents of the type that induces the highest long-run growth rate will grow, until, in the long run, almost all agents are of this type.

In Section \ref{sec:Risk-Attitude} we show that our key result implies
that the type with the highest growth rate is (1) risk
averse with respect to most lotteries over consumption (due to the
concavity of the function $\psi$ relating consumption and fertility),
but (2) risk loving with respect to sufficiently positively skewed
lotteries. Since biological types evolve slowly, it is likely that this risk attitude persists in modern times, even though the birth rate may no longer be increasing in the consumption
rate. This finding fits the stylized empirical fact that people,
although being in general risk averse, are skewness loving. That is, people like lotteries involving a small probability
of winning a high prize. (See, for example, \citealp{golec1998bettors,garrett1999gamblers}.)

\paragraph{Structure}

The rest of the paper is organized as follows. Section \ref{sec:Motivating-Example} informally presents the essence
of our key result. The model is presented in Section \ref{sec:Model}.
Section \ref{sec:key-result} formally presents our {key} result. In Section
\ref{sec:Risk-Attitude} we discuss the implications of our result
for attitudes to risk. Section \ref{sec:Dynasties-and-Structured}
{extends our baseline model by allowing dependency between redraws of heritable risk within each of a number of dynasties, with independence across dynasties, which seems plausible in various applications. We show that this extension does not affect our results for infinite populations. By contrast, this structure can affect the growth rate of finite populations, which we investigate by numerical simulations. We discuss several additional related references in Section \ref{sec:Discussion} and conclude in Section \ref{sec:Conclusions}.}

\section{Informal Treatment of Key Result\label{sec:Motivating-Example}}

The following example conveys the gist of our key result. Consider
three populations, each having a random fertility rate (which is independent
of the agent's age) with the same marginal distribution. Each population has a probability $q_{\ell}$
of having a low fertility rate of $r_{\ell}$, and a probability  $q_{h}=1-q_{\ell}$ of having
a high fertility rate of $r_{h}$. For notational compactness, we
now take as implicit the dependence of fertility on consumption rates
$c_{i},i=\ell,h$. For simplicity, we focus on fertility, so that
there is no mortality. The source of risk is independent across populations.

In Population 1 risk is idiosyncratic; that is, the fertility rate
of each agent is independent of the fertility rate of all other agents in the populations and, in particular, of her parent's fertility rate. Applying the law of large numbers, the number
of agents in Population 1 at time $t$ is equal to $N(t)=e^{\left(q_{\ell}\cdot r_{\ell}+q_{h}r_{h}\right)\cdot t}$,
where $N(0)=1$, and the annual growth rate is $\frac{1}{t}\cdot\ln N(t)=q_{\ell}\cdot r_{\ell}+q_{h}r_{h}\equiv\mu$.

In Population 2 risk is aggregate. There are two states: $\ell$ and
$h$. In state $\ell$, all agents have fertility rate $r_{\ell}$,
and in state $h$, all agents have fertility rate $r_{h}$. There
is a continuous probability rate $\lambda$ that the state is
redrawn. If it is, the fertility rate is $r_{\ell}$ with probability $q_{\ell}$
and $r_{h}$ with probability $q_{h}$. What is the {(long-run)} growth rate of
the population exposed to this aggregate risk? If $N(t)$ is the population
at time $t$, and $N(0)=1$, it follows that
\[
\frac{\ln N(t)}{t}=\frac{r_{\ell}\cdot({\text{time in state }}\ell)+r_{h}\cdot(\text{time in state }h)}{t}\longrightarrow q_{\ell}\cdot r_{\ell}+q_{h}\cdot r_{h}=\mu,
\]
as ${t\rightarrow\infty}$, given the evident ergodicity of the process. Thus, as
shown in \cite{robatto2017biological}, both idiosyncratic risk
and aggregate risk induce the same growth rate.

We introduce a novel form of risk in Population 3, called heritable risk. Each
agent redraws her heritable birth rate independently of all other
agents at a rate $\lambda$, and at each redraw the agent
gets a fertility rate $r_{\ell}$ or $r_{h}$ with probability $q_{\ell}$
or $q_{h}=1-q_{\ell}$, respectively
{(independently of all other events). The previous literature makes
an implicit assumption that each offspring is given a fresh draw, and so all offspring are equivalent and evolutionary success entails simply counting these undifferentiated offspring. By contrast,
suppose that each offspring inherits the 
{\emph{realized} fertility rate} of the parent. Since offspring are now differentiated, the value of these 
offspring varies with type and simply counting them is inadequate. Our key result shows that in this case the growth rate is
strictly
higher than the expectation $\mu$, and indeed converges to $r_{h}$ as $\lambda\rightarrow 0$.}

To understand the gist of the argument, consider a simplified alternative
setup in which redraws arrive deterministically and in synchrony every
$\tau$ periods,
which is comparable to an arrival rate of $\lambda=1/\tau$.
As before, the redrawn values of different agents are independent.
On each draw, a share $q_{\ell}$ of the agents get $r_{\ell}$ and
the remaining agents get $r_{h}$. If the initial population is of
size $1$, then, after a time $k\cdot\tau$, the population is $N(k\cdot\tau)=(q_{\ell}\cdot e^{r_{\ell}\cdot\tau}+q_{h}\cdot e^{r_{h}\cdot\tau})^{k},$
so that
\[
\frac{1}{k\cdot\tau}\ln N(k\cdot\tau)=\frac{1}{\tau}\cdot\ln(q_{\ell}\cdot e^{r_{\ell\cdot}\tau}+q_{h}\cdot e^{r_{h}\cdot\tau})\equiv\bar{g}(\lambda).
\]
It follows that the growth rate of the population, $\bar{g}\left(\lambda\right)$,
is decreasing in $\lambda$, $\bar{g}(\lambda)\rightarrow r_{h}$
if $\lambda\rightarrow0$ ($\tau\rightarrow\infty$), and $\bar{g}(\lambda)\rightarrow\mu\equiv q_{\ell}\cdot r_{\ell}+q_{h}\cdot r_{h}$,
if $\lambda\rightarrow\infty$ ($\tau\rightarrow0$). This, in particular,
implies that the growth rate is strictly higher than the lottery's
expectation $\mu$, which is the growth rate induced by either idiosyncratic
risk or aggregate risk with the same marginal distribution.

{Recall that $x^{*}$ is the growth rate in the general model. What Theorem \ref{thm:main-result} shows, more precisely, is that $x^{*}>\mu$ and $x^{*}> r_h-\lambda$.\footnote{The simplifying assumption that the intervals between redraws are deterministic (rather than stochastic intervals induced by a Poisson process) decreases the growth rate, and thus the above example
might yield a lower growth rate than the lower bound $r_h-\lambda$ of Theorem 1.} This latter result implies that $x^{*}\rightarrow r_h$ as $\lambda\rightarrow 0$, given $x^{*}<r_h$. Figure \ref{fig:illus} illustrates our result for the values $r_h=5\%$, $r_l=0\%$, and $q_h=10\%$; i.e., for a binary lottery that yields a high annual birth rate of $5\%$ with probability $10\%$ and a zero birth rate with probability $90\%$. When risk is either idiosyncratic or aggregate the (long-run) growth rate is equal to the expected birth rate $\mu=0.5\%$. Theorem \ref{thm:main-result} (and the informal argument above) shows that when the risk is heritable the growth rate is strictly larger than $\mu$. The figure also draws the exact growth rate induced by heritable risk according to the explicit formula presented in Claim \ref{cla:binary} (in Appendix \ref{sec:Explicit-Solution-for}) for binary lotteries. As can be seen from the figure, when the redraw rate $\lambda$ is very small (resp., large) with respect to  $r_h$, then the growth rate is slightly above $r_h-\lambda$ (resp., $\mu$).}
\begin{figure}
\begin{centering}
\includegraphics[scale=0.44]{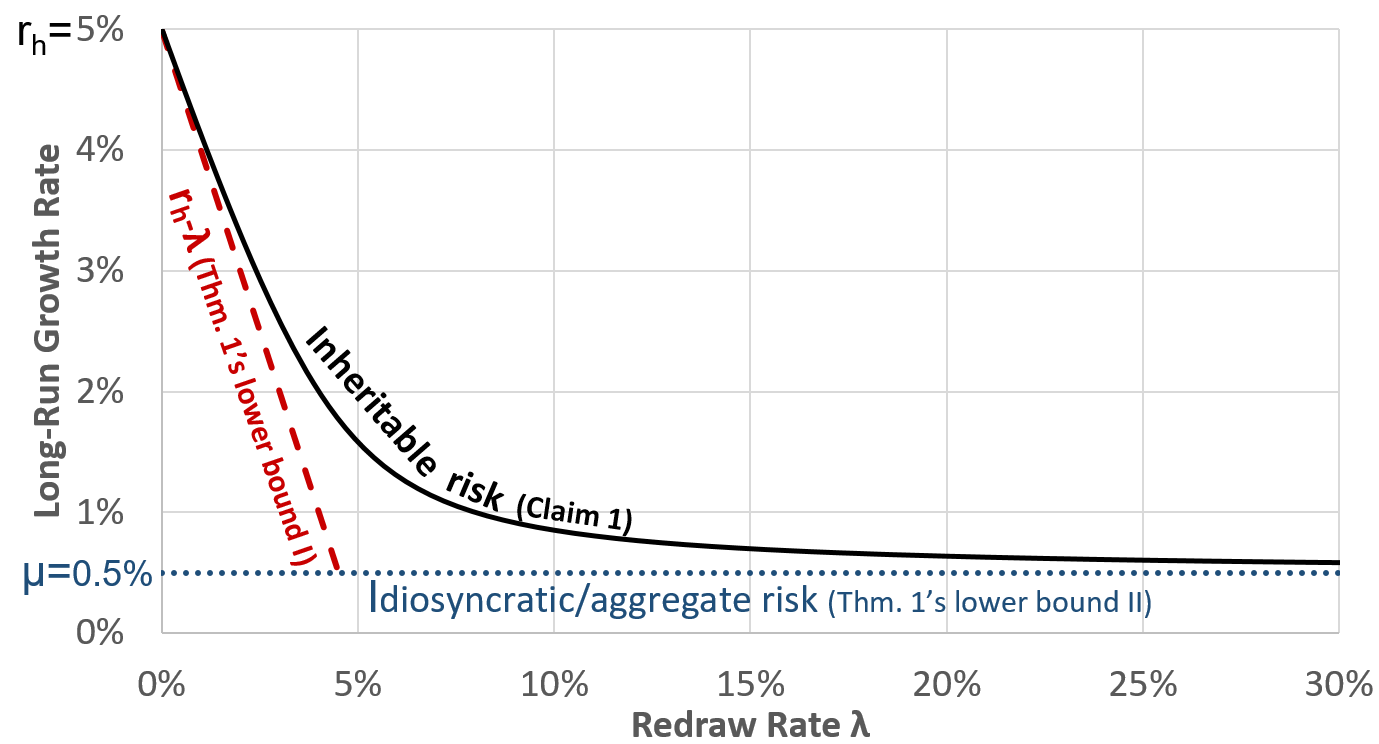}
\par\end{centering}
\caption{\label{fig:illus}{Long-Run Growth Rate for a Binary Lottery ( $r_h=5\%$, $r_l=0\%$, and $q_h=10\%$).}}
\end{figure}

\section{Model\label{sec:Model}}

Consider a continuum population of an initial mass one. Time is continuous,
indexed by $t\in\mathbb{R}^{+}$.
To simplify matters, we assume that reproduction is asexual.  The growth process depends on the parameters
$\left(\delta,\left(X,q_{x},\lambda_{x}\right),\left(Y,q_{y},\lambda_{y}\right),\left(Z,q_{z},\lambda_{z}\right)\right)$,
as described  below.

In what follows, we first present an intuitive description of Poisson processes on the individual level that incorporate the probability of each agent dying, giving birth, and changing her birth rate 
{(parts (i) below)}. We then specify the corresponding exact evolution of the large population that is assumed in our model 
{(parts (ii)} below).\footnote{The formalization of the intuitive claim that the idiosyncratic Poisson process for the birth rate of an individual in a large population implies the mean is exactly attained raises various technical difficulties. See \cite{duffie2012exact} (and the citations
therein) for details.}

\begin{enumerate}
\item (i) We suppose intuitively that each agent experiences a constant Poisson {death rate} $\delta\geq0$ that is \emph{independent} of all other random variables and, in particular, of all components of the birth rates.

(ii) We assume precisely that, in each infinitesimal period of time between $t$ and $t+dt$, a fraction $\delta\cdot dt$ of the population dies,
where this fraction is uniform across all components of the birth rate.
\end{enumerate}
Each individual $i$  at time $t$ has a birth rate $\boldsymbol{b}^{i}\left(t\right)=\boldsymbol{x}^{i}\left(t\right)+\boldsymbol{y}^{i}\left(t\right)+\boldsymbol{z}\left(t\right)$ with three components. These components are constructed as follows:
\begin{enumerate}
\item[2.]  (i) The random variable $\boldsymbol{x}^{i}\left(t\right)\geq0$ is the heritable component of the birth rate. A newborn agent obtains the heritable birth rate of
her parent. We assume that the random variable $\boldsymbol{x}^{i}$
has a finite support $X=supp\left(\boldsymbol{x}\right)=\left\{ x_{1},...,x_{n}\right\} $,
where $x_{1}<...<x_{n}$ 
{and $n\geq2$}. The function $q_x:X\rightarrow (0,1),~ \sum_{x\in X} q_x =1,$ assigns a probability to each $x\in X.$ 
Intuitively, {in each infinitesimal period of time $dt$} each agent has a probability of $\lambda_{x}\cdot dt$ 
of redrawing her heritable birth rate
{(where $\lambda_x>0$)}, and
these redrawing events are independent of all other events.

(ii) The precise assumptions on the heritable component are as follows. Suppose that $w(t)$ is the total population at time $t$ and $w_k(t)$ is the {mass} of agents who
are endowed with heritable component $x_k$. Then the rate of increase of $w_k(t)$ is
\begin{equation}
\frac{dw_{k}(t)}{dt}=w_{k}(t)x_{k}-\lambda_{x}w_{k}(t)+\lambda_{x}w(t)q_{x}(x_{k})-\delta w_{k}(t).\label{eq:w_k_t}
\end{equation}

The first term expresses the increase in $w_k(t)$ due to offspring who are endowed with 
$x_k$. { This captures the  key characteristic of heritable risk that all offspring are endowed with the same component $x_k$ as their parent. Since this term is independent of $\lambda_x$, it will follow that $w_k(t)$ grows at rate $x_k-\delta$ when $\lambda_x\rightarrow 0$.} The second term expresses the loss from $w_k(t)$ of those agents who redraw. The third term represents the increase due to all agents from $w(t)$ (including those from $w_k(t)$) who redraw and obtain $x_k$. The final term represents the loss from $w_k(t)$ due to death.

\item [3.] (i) The random variable $\boldsymbol{y}^{i}(t)\geq0$ is the idiosyncratic
component of the birth rate. The idiosyncratic birth rate of an agent is independent of all other random variables governing the birth rates in the population. The random variable $\boldsymbol{y}^{i}$ has a finite support
$Y=supp\left(\boldsymbol{y}\right)=\left\{ y_{1},...,y_{n_{y}}\right\} $. 
The function $q_y:Y\rightarrow (0,1],~ \sum_{y\in Y} q_y =1$, assigns a probability to each $y\in Y.$ 
In {each infinitesimal period of time $dt$} each agent has a probability of $\lambda_{y}\cdot dt$ of redrawing her idiosyncratic birth rate, and
these redrawing events are independent of all other events.

(ii)
The precise assumption is that the idiosyncratic component within any group of agents always reflects the
distribution $q_y$. That is, the share of agents with idiosyncratic outcome $y_{\ell}$, for example, in the group of agents with heritable outcome $x_k$ is exactly equal to $q_y(y_{\ell})$ for any time
$t \geq 0$. This implies that the idiosyncratic component in any group of agents with any heritable component $x_k$ is exactly equal to the expectation $\mu_y$.

\item[4.] The aggregate component of the birth rate $\boldsymbol{z}\left(t\right)\geq0$ can be handled more straightforwardly since all agents in the population share this aggregate
rate.
We assume that the random variable $\boldsymbol{z}^{i}$ has a finite support $Z=supp\left(\boldsymbol{z}\left(t\right)\right)=\left\{ z_{1},...,z_{n_{z}}\right\} $. The function
$q_z:Z\rightarrow (0,1],~ \sum_{z\in Z} q_z =1,$ assigns a probability to each $z\in Z.$ 
At time $t=0$ the aggregate birth rate $\boldsymbol{z}\left(0\right)$ is randomly determined according to the distribution $q_{z}$.
In each infinitesimal period of time between $t$ and $t+dt$ a new random value of the aggregate birth rate is drawn independently
(according to $q_{z}$) with a probability of $\lambda_{z}\cdot dt$, where $\lambda_{z}>0$. This aggregate birth rate applies to all individuals
in the entire population equally.

\end{enumerate}

\section{Key Result\label{sec:key-result}}

Let ${w}\left(t\right)$ denote the mass of the population
at time $t$. We normalize ${w}\left(0\right)=1$. We say
that the growth process of ${w}\left(t\right)$ given by
$\left(\delta,\left(X,q_{x},\lambda_{x}\right),\left(Y,q_{y},\lambda_{y}\right),\left(Z,q_{z},\lambda_{z}\right)\right)$
has an \emph{equivalent (long-run) growth rate} $g\in\mathbb{R}$
if and only if
\[
lim_{t\rightarrow\infty}\frac{\ln{w}\left(t\right)}{t}=g,{\text{ almost surely.}}
\]

Let $\mu_{x}=\sum_{k}x_{k}\cdot q_{x}\left(x_{k}\right)$ (resp.,
$\mu_{y}=\sum_{k}y_{k}\cdot q_{y}\left(y_{k}\right)$, $\mu_{z}=\sum_{k}z_{k}\cdot q_{z}\left(z_{k}\right)$)
be the expectation of the heritable (resp., idiosyncratic, aggregate)
birth rate.
We show that the equivalent growth rate is the sum of four components:
$g=f\left(X,q_{x},\lambda_{x}\right)+\mu_{y}+\mu_{z}-\delta.$ The
results on the idiosyncratic and aggregate components of the overall
growth rate accord with the existing literature (\citealp{robatto2017biological}),
namely, these components are equal to $\mu_{y}$ and $\mu_{z}$, respectively.
The novel part of the result is that the heritable birth component
satisfies
\[
f\left(X,q_{x},\lambda_{x}\right)\in\left(\max\left(\mu_{x},x_{n}-\lambda_{x}\right),x_{n}\right).
\]
That is, the heritable birth component is always larger than $\mu_{x}$,
and it cannot be more than $\lambda_{x}$ away from the highest realization
$x_{n}$. The first property shows that the desirability of heritable risk
is that it induces a higher growth rate than comparable aggregate
or idiosyncratic risk. The second property shows that the highest
realization of the heritable risk has a substantial influence, regardless
of how low is its probability. That is, a lottery in which $x_{n}>\lambda_{x}$
induces a growth rate of at least $x_{n}-\lambda_{x}$ regardless
how small $q_{x}(x_n)$ and $\mu_{x}$ might be.

The intuition is that the distribution of the heritable birth rate
in the population converges to a distribution $p\in\Delta\left(X\right)$
that first-order stochastically dominates $q_{x}$. This is because,
at each point in time, agents with a high heritable birth rate tend
to have more offspring and these offspring share the parent's heritable
birth rate. Hence, in a steady state, the share of agents with a high
heritable birth rate is strictly higher than $q$. Higher values of
$\lambda$ reduce this effect, as the offspring redraw more rapidly a new value
for their heritable birth rate (according to $q_{x}$).

The final claim is that $f\left(X,q_{x},\lambda_{x}\right)$
increases following a mean-preserving spread of the heritable birth
rate. The intuition is that a mean preserving spread increases the
high $x_{k}$'s while decreasing the low $x_{k}$'s, and there is
a net gain from this due to the over-representation of high $x_{k}$-agents
in the steady-state distribution.
\begin{thm}
\label{thm:main-result}Let $\left(\delta,\left(X,q_{x},\lambda_{x}\right),\left(Y,q_{y},\lambda_{y}\right),\left(Z,q_{z},\lambda_{z}\right)\right)$
be a growth process. Then its equivalent growth rate is equal to $g=f\left(X,q_{x},\lambda_{x}\right)+\mu_{y}+\mu_{z}-\delta,$
where, setting $f\left(X,q_{x},\lambda_{x}\right)=x^{*}$ for compactness,
$x^{*}$ is the unique positive solution of
\[
x^{*}=\lambda_{x}\sum_{k=1}^{n}\frac{q_{k}\cdot x_{k}}{\lambda_{x}+x^{*}-x_{k}}\in\left(\max\left(\mu_{x},x_{n}-\lambda_{x}\right),x_{n}\right),
\]

with $q_{k}\equiv q_x\left(x_{k}\right)$ for each $k\in\left\{ 1,..,n\right\} $. {It follows that $x^{*}\rightarrow \sum_{k=1}^n q_k x_k=\mu_x$ as $\lambda_x\rightarrow\infty$ and $x^{*}\rightarrow x_n$ as  $\lambda_x\rightarrow 0.$}

Moreover, if $\left(X',q'_{x}\right)$ is a mean-preserving spread
of $\left(X,q_{x}\right)$, then $f\left(X',q_{x'},\lambda_{x}\right)>f\left(X,q_{x},\lambda_{x}\right)$.
\end{thm}
\begin{proof}[Sketch of proof; The full proof is in Appendix \ref{subsec:Proof-of-Theorem-one}]
Since the novel result here concerns heritable risk, let us suppose,
for simplicity, that there is no aggregate risk, idiosyncratic risk,
or mortality. Suppose further that the size of the population at time
$t$ is $w(t)$ and that a steady-state fraction $p_{k}$ of this
population has birth rate $x_{k}$.\footnote{The formal proof deals with the general case, and shows global convergence to the steady state.}  The net increase in each infinitesimal period $dt$ of
those agents with birth rate $x_{k}$ is then $p_{k}\cdot x_{k}\cdot w\left(t\right)\cdot dt$
(offspring born to parents with a birth rate $x_{k}$ who inherit
this rate) minus $(p_{k}-q_{k})\cdot\lambda_{x}\cdot w\left(t\right)\cdot dt$.
(Note that $\lambda_{x}\cdot w\left(t\right)\cdot dt$ agents have
redrawn a fresh value for the heritable birth rate, and the share
of $x_{k}$-agents among them has changed from $p_{k}$ to $q_{k}$.)
The increase in the total mass of agents is $\sum_{k}p_{k}\cdot x_{k}\cdot w\left(t\right)\cdot dt$
(the sum of offspring born to parents with each birth rate). The equilibrium
value of $p$ should match the ratio of the net increase of agents
with a high heritable birth rate to the net increase of the population,
such that
\[
p_{k}=\frac{\left(p_{k}\cdot x_{k}+(q_{k}-p_{k})\cdot\lambda_{x}\right)\cdot w\left(t\right)\cdot dt}{\sum_{k}p_{k}\cdot x_{k}\cdot w\left(t\right)\cdot dt}=\frac{p_{k}\cdot x_{k}+(q_{k}-p_{k})\cdot\lambda_{x}}{\sum_{k}p_{k}\cdot x_{k}\cdot}.
\]

Solving for $p_{k}$ yields (where $x^*\equiv\sum_{k}p_{k}\cdot x_{k}$):
\begin{equation}
p_{k}=\frac{\lambda_{x}\cdot q_{k}}{\lambda_{x}+x^*-x_{k}}.\label{eq:p-main}
\end{equation}
This solution assumes that $p_{k}$ is positive for all $k$ so that
$x^*>x_{n}-\lambda_{x}.$ Next we multiply each $k$-th equation
by $x_{k}$ and sum to an equation in one unknown:

\begin{equation}
x^*=\sum_{k}\frac{x_{k}\cdot\lambda_{x}\cdot q_{k}}{\lambda_{x}+x^*-x_{k}}.\label{eq:x-star-main}
\end{equation}
Observe that in the domain $x^*>x_{n}-\lambda_{x}$
the LHS (resp., RHS) is increasing (resp., decreasing) in $x^{*}$,
which implies that there exists a unique solution $x^*>x_{n}-\lambda_{x}$
to Eq. (\ref{eq:x-star-main}). Substituting this solution in Eq.
(\ref{eq:p-main}) yields the unique steady-state distribution $p$. {From Eq.  (\ref{eq:x-star-main}) it follows that} \begin{equation}
x^*=\sum_{k}\frac{x_{k}\cdot q_{k}}{1+\frac{x^*-x_{k}}{\lambda_x}} \rightarrow\sum_{k}x_{k}\cdot q_{k},\label{eq:x-star-sub}
\end{equation}
 {as $\lambda_x\rightarrow\infty.$ Since $x^*\in (x_n-\lambda_x,x_n),$ it is also immediate that $x^*\rightarrow x_n$ as $\lambda_x\rightarrow 0.$}

The final claim is proved as follows. Eq. (\ref{eq:x-star-main})
can be written as
\begin{equation}
\boldsymbol{E}_{\boldsymbol{x}}\left[\frac{\boldsymbol{x}\cdot\lambda}{\lambda+x^*-\boldsymbol{x}}\right]=x^*,\label{eq:x*-expect}
\end{equation}
where $\boldsymbol{x}$ is the random variable $\left(X,q_{x}\right)$.
The fact that $\frac{\boldsymbol{x}\cdot\lambda}{\lambda+x^*-\boldsymbol{x}}$
is a convex function of $\boldsymbol{x}$ implies that it increases
following a mean-preserving spread. {This, together with the fact that it is decreasing in $x^{*}$,} implies that in
order to maintain Eq. (\ref{eq:x*-expect}) following a mean-preserving
spread, the growth rate $x^{*}$ must increase.
\end{proof}

\section{Risk Attitude\label{sec:Risk-Attitude}}

We suppose that individuals in a large population may have different
types, where the type represents the agent's risk attitude---in particular, how
the agent chooses between a risky consumption option and a safe one. An agent has the same
type as her parent. Occasionally, new types
may be introduced into the population as genetic mutations.
Observe that the population share of agents that are endowed with
the type that induces the highest long-run growth rate for its practitioners
will grow, until, in the long run, almost all agents are of this type.
{For example, suppose that there are two types $\theta,\theta'$
in the population, each with an initial frequency of 50\% that induce growth rates $g(\theta), g(\theta')$, respectively. After
time $t$ the share of agents having type $\theta$ will be $\frac{e^{g\left(\theta\right)t}}{e^{g\left(\theta\right)t}+e^{g\left(\theta'\right)t}},$ which converges to one as $t\rightarrow\infty$, if $g\left(\theta\right)>g\left(\theta'\right)$.
See \cite{robson2011evolutionary} and the citations therein,
for a more detailed argument of why natural selection induces agents
to have types that maximize the long-run growth rate.}

Now consider a setup in which agents face choices between various
alternatives, where each alternative corresponds to a lottery over
the consumption rate. We assume that the birth rate is a concave increasing
function of consumption, given by $\psi:\boldsymbol{R}^{+}\rightarrow\boldsymbol{R}^{+}$.
To simplify the presentation, assume that the birth rate is entirely
heritable; the result remains qualitatively the same if the birth
rate induced by consumption has all three risk components (heritable, idiosyncratic,
and aggregate). We now argue that a growth-rate-maximizing
type induces agents
(1) to be risk averse with respect
to most lotteries over consumption,
and, yet, (2) to strictly prefer
some fair lotteries that are sufficiently skewed. Thus,
natural selection should induce agents to have a risk attitude combining
risk aversion and skewness loving.

For simplicity, assume that an agent faces choices among lotteries over consumption $\left(C,q\right)$ with a finite support $C$, where $c>0$ for all $c\in C.$ Suppose probabilities are assigned by $q:C\rightarrow[0,1],\sum_{c\in C}q(c)=1$. Let $m=\max\left\{ c\in C\right\} $ be the maximal possible realization and let
$\bar{c}=\sum_{c\in C}q(c)\cdot c$ be the mean. 
For any fixed lottery, we show that, once $\psi$ is sufficiently concave, the constant consumption rate of $\bar{c}$ will induce a higher long-run
growth rate than the lottery $\left(C,q\right)$. This explains why
the growth-rate-maximizing type should induce the agents to be risk
averse with respect to most lotteries, when $\psi$ is sufficiently concave. Consider, for example, the function $\psi(c)=c^{\beta}$ for $\beta \in (0,1]$. Theorem \ref{thm:main-result} shows that the individual prefers the lottery $\left(C,q\right)$ to the mean $\bar{c}$ when $\beta =1$ so that $\psi(c)=c$. However, if $\beta$ is small enough this preference is reversed. This is formalized in the following 
proposition
that shows that, given any lottery over consumption, the individual will prefer the mean consumption to the lottery if $\beta$ is small enough.

\begin{proposition}\label{smallbeta}
 {Suppose that $\psi(c)=c^{\beta}$ for $\beta \in (0,1]$. Then,} given any gamble $\left(C,q\right)$, the mean $\bar{c} =\sum_{c\in C}q\left(c\right)\cdot c$ induces a higher growth rate than the lottery $(C,q)$, if $\beta>0$ is close enough to $0$.
\end{proposition}

\noindent \textbf{Proof} See Appendix \ref{subsec:Proof-of-Lemma}.

On the other hand, for a fixed function $\psi(c)$, if the lottery $\left(C,q\right)$ is sufficiently skewed---i.e., if
$m$ is high enough and $q\left(m\right)$ is low enough so that $\psi(m)-\lambda_{x}>\psi(\bar{c})$---then the
lottery induces a strictly higher growth rate than the constant consumption
rate of $\bar{c}$. This follows from Theorem \ref{thm:main-result} since the lottery's
long-run growth rate is bounded from below by $\psi(m)-\lambda_{x}$.
This implies that growth-rate-maximizing agents
would prefer a sufficiently positively skewed lottery to its expectation.

{The above argument suggests that natural selection has
induced people to be generally risk averse and sometimes skewness
loving. As biological types evolve slowly, it seems likely that this
risk attitude persists in modern times, in which, arguably, the
birth rate is no longer increasing in the consumption rate. Thus,
our findings fit the stylized fact that people, although being in
general risk averse, are skewness loving, in the sense of being risk
loving with respect to lotteries involving a small probability of
winning a high prize (e.g., buying state lottery tickets; see
\citealp{golec1998bettors,garrett1999gamblers}). }

\section{ {An Extended Model: Dynasties}\label{sec:Dynasties-and-Structured}}

In our baseline model, the event of an agent redrawing her heritable
birth rate is independent of her parent's redrawing event.
In various environments, it seems plausible that members of a  {dynasty} may change their
heritable birth rate together,  while remaining independent of other dynasties. For example, if heritable risk is induced
by a foraging technique or  {a geographical location}, and environmental changes affect the effectiveness
of the foraging technique, then an entire dynasty of agents  {(who use the same foraging technique or live in the same geographical location)} may simultaneously change their heritable birth rate. 

In this section we extend our baseline model by introducing dynasties, and allowing dependency between redraws of heritable risk within each dynasty. We show that this extension does not affect our results for infinite populations. By contrast, this structure can affect the growth rate of finite populations, which we investigate by numerical simulations. 

\paragraph{Extended model}
 In what follows we extend our baseline model to a continuum of dynasties. We adopt the same notation as in the baseline model. The processes according to which agents die, are born, and change their idiosyncratic and aggregate birth components remain the same as in the baseline model. Importantly, each offspring is born into the same dynasty as her parent.

 {Let $[0,1]$ be the set of dynasties, where each agent $i$ in the initial population (of mass one) lives in a different dynasty $i\in[0,1]$. Each dynasty is initially endowed with a heritable birth rate according to the distribution $q_x$. Formally, we assume that the mass of dynasties having heritable birth component $x_k\in X$ is equal to $q_x(x_k).$ The heritable birth component of each agent is tied to the heritable birth component of all members of her dynasty.}

 {There are two processes that change the heritable birth component of agents. We begin with an intuitive description of two Poisson processes that change the heritable birth  component: migration and a dynasty's 
 redraw. We then specify the corresponding exact evolution of the distribution of the heritable birth component as the product of these two processes.}

\begin{enumerate}
\item  {\emph{Migration}:}  
 {Intuitively, in each infinitesimal time $dt$ each agent has a probability of $\lambda_m\cdot dt$ (where $\lambda_m\geq0$) to leave her dynasty and move to a new random dynasty (distributed uniformly in the set of all dynasties [0,1]). These migration events are independent of all other events. Following the migration, the agent is endowed with the heritable birth component of her new dynasty.} 

\item  {\emph{Dynasty's redraw}:}  {Intuitively, in each infinitesimal time $dt$ each dynasty $j\in[0,1]$ has a probability of $\lambda_r \cdot dt$ to redraw a fresh value for its heritable birth component. These redrawing events are independent of all other events. When a dynasty redraws its heritable component it changes the heritable component of all agents living in that dynasty.}
\end{enumerate}

 {Next, we formulate the precise dynamics of the mass of agents $w_k(t)$ who are endowed with the heritable component $x_k$ that is induced by the combined effect of migration and a dynasty's redraws. The rate of increase of $w_k(t)$ is}

\begin{equation}
\frac{dw_{k}(t)}{dt}=w_{k}(t)x_{k}-\lambda_{m}w_{k}(t)+\lambda_{m}w(t)q_{x}(x_{k})-\lambda_{r}w_{k}(t)+\lambda_{r}w(t)q_{x}(x_{k})-\delta w_{k}(t).\label{eq:w_k_t_ex}
\end{equation}

 {The first and final terms are identical to Eq. (\ref{eq:w_k_t}) of the baseline model. The first term expresses the increase in $w_k(t)$ due to offspring who are endowed with $x_k$. The final term represents the loss from $w_k(t)$ due to death.}

 {The second and third terms express the impact of migration. The second term ($-\lambda_{m}w_{k}(t)$) is the loss from $w_k(t)$ of agents who migrate out of dynasties with heritable component $x_k$. The third term ($\lambda_{m}w(t)q_{x}(x_{k})$) represents the increase due to all agents from $w(t)$ (including those from  $w_k(t)$) who migrate into dynasties with heritable component $x_k$.}

 {The fourth and fifth terms express the impact of redraws of dynasties. The fourth term ($-\lambda_{r}w_{k}(t)$) represents the loss from  $w_k(t)$ of agents who live in dynasties with heritable component $x_k$ that redraw a fresh draw. Finally, the fifth term ($\lambda_{r}w(t)q_{x}(x_{k})$) represents the increase due to all agents from dynasties (including dynasties that already had $x_k$) that draw a fresh value of heritable component $x_k$.}

 {Observe that Eq. (\ref{eq:w_k_t_ex}) is equivalent to Eq. (\ref{eq:w_k_t}) of the basic model except that $\lambda_x$ is replaced with $\lambda_m+\lambda_r$. That is, the dynamics of $w_k(t)$ in the extended model is exactly the same as in the baseline model with $\lambda_x=\lambda_m+\lambda_r$. As the impact of the heritable component on the growth rate is exactly captured by $w_k(t)$, this implies that all of our results hold in this extended setup with dynasties.}
\paragraph{Ever-growing population with dying dynasties}

Consider a simple case in which: (1) $\lambda_{m}=0$, i.e., agents
never migrate, and each dynasty is an isolated subpopulation, (2)
all risk is heritable, (3) the growth rate predicted by the continuum
model is positive, and (4) an aggregate birth rate with the same marginal
distribution induces a negative growth rate. For example, assume
that the heritable birth rate of each dynasty is randomly chosen to be
either $x_{l}=0\%$ or $x_{h}=2\%$ with equal probability, that there
is a constant death rate of
$\delta=1.4\%$, and that the redrawing rate
of the heritable risk by each dynasty is given by $\lambda_{r}=2\%$.
Theorem \ref{thm:main-result}
and Claim \ref{cla:binary} imply that this heritable birth rate induces a positive growth rate
of 0.014\%, while if the birth rate were induced by aggregate risk
with the same distribution, then the growth rate would be negative:
$-0.4 \% =\left(0.5\cdot0\%+0.5\cdot2\%\right)-1.4\%$.

Each dynasty is a completely isolated subpopulation with
risk that is essentially aggregate within the subpopulation. Thus, each dynasty
is doomed to extinction since it has a negative growth rate of $-0.4\%$.
This yields a seemingly paradoxical result: the entire population
grows exponentially, while each of its dynasties eventually becomes
extinct. Such a result holds with a continuum of dynasties.
Although each dynasty eventually dies, in each finite time there is
still a continuum of surviving dynasties with a large realized growth
rate, such that the growth rate of the entire population can be positive.

The intuition behind this result can be {illustrated}
more clearly
in a simple {alternative} setup in which each dynasty in each period can be either
successful or go extinct with equal probability. A successful dynasty
increases its size by a factor of 4 in each period. Observe that the
expected size of each dynasty after $t$ period is $2^{t}$, which
is the product of a tiny probability of $0.5^{t}$ of the dynasty
surviving and the very large size of the dynasty ($4^{t}$), conditional
on surviving. If the population includes a continuum of mass one of
dynasties, then (by applying an exact law of large numbers) after
$t$ periods the size of the population is $2^{t}$ (with probability
one), and this population is concentrated on a continuum of a small
mass of $0.5^{t}$ of surviving large dynasties. Thus, the population's
size converges to infinity, even though the share of surviving dynasties
converge to zero. By contrast, if the number of dynasties were finite
(instead of a continuum), then after a sufficiently long finite time,
the population's size would eventually be zero with probability one.

\paragraph{Finite populations}

The result of an ever-growing population in which each dynasty is
eventually doomed cannot happen when the number of dynasties is finite.
Since each dynasty is doomed to extinction, so too
is the overall population. However, the fact that the mean size of
each subpopulation is growing implies that the overall population may grow
significantly in the interim. As the finite model converges to the
continuum model, this initial growth phase becomes more and more prolonged,
and the inevitable ultimate demise of the population is postponed
indefinitely.

When there is no migration, a large 
{finite} population tends to
ultimately put all its eggs in one basket. That is, the distribution
of the finite population over its subpopulations tends to become very
unequal, often concentrated in just one subpopulation. Such large
subpopulations hold up the mean, which is the growth rate found here. Once
the population is concentrated like this, however, doom is inevitable
because the 
heritable risk of a large subpopulation, essentially, becomes
an aggregate risk since it affects a large share of the entire population.

Migration introduces a new element to these observations. In the finite
model migration has a distinct effect from that of the redraw rate.
If some subpopulations grow large, and others shrink, migration acts
to redistribute the population. This means that the population can
exploit the numbers in the large subpopulations, while diversifying
the risk. These observations motivate the simulations described below.

\subsection{Numerical Analysis of Finite Populations\label{sec:Numeric-Analysis-of}}

In this section we present simulations that test whether our theoretical
results for continuum populations hold for finite populations.

\paragraph{Description of the Simulation}
The simulation is a discrete-time  version of the extended model (with dynasties) described above. Specifically, the basic time step of the simulation is one year, and we replace each continuous Poisson rate with the respective independent per-year probability (e.g., an annual birth rate of $2\%$ is replaced with an independent probability of $2\%$ of each agent giving birth in each year). The
Python code (contributed by Renana Heller) is included in the online supplementary material.

We describe here the results of 150 simulation runs, which comes from
15 runs of 10 different parameter combinations. In each simulation
run, the initial population includes 3,000 agents that are initially
randomly allocated to 300 dynasties. The aggregate birth rate and the
idiosyncratic birth rate are both equal to zero (i.e., $\mu_{y}=\mu_{z}=0$).
The {heritable}
birth rate in each dynasty is randomly chosen to be either
$x_{l}=0\%$ or $x_{h}=2\%$ with equal probabilities (i.e., $q=0.5$).
We set the total annual rate at which each agent switches the {heritable} birth rate to be $\lambda_{m}+\lambda_{r}=2\%$. We set the annual
death rate at 1.4\%, which implies that the theoretical prediction
for a continuum population (see Claim \ref{cla:binary} in Appendix
\ref{sec:Explicit-Solution-for}) is that: (1) the share of agents
with a high {heritable} birth rate converges to about $71\%$, and (2) the
annual long-run growth rate will be about $0.014\%$. A naive prediction
that treats {heritable} risk as if it were aggregate risk predicts a long-run
growth rate of $-0.4\% =\left(0.5\cdot0\%+0.5\cdot2\%\right)-1.4\%$.
Due to technical constraints and time limits we stopped each simulation
run after (1) 20,000 years have passed, (2) the population size increases
by 300-fold to 1,000,000 or more, or (3) the population size decreases
by 300-fold to 10 or less (henceforth, \emph{extinction}). The various
simulation runs study 10 different ratios $\frac{\lambda_{m}}{\lambda_{r}}$
of the migration rate relative to the {dynastic} risk redrawing rate (while
maintaining $\lambda_{m}+\lambda_{r}=2\%$): 0.01, 0.02, 0.05, 0.1,
0.25, 0.5, 1, 2, 4, 10.

\paragraph{Numerical Results}

Figure \ref{fig:Sample-Results-for} presents four representative
simulation runs with ratios: 0.01 ($\lambda_{m}=0.02\%$, $\lambda_{r}=1.98\%$),
0.05 ($\lambda_{m}=0.1\%$, $\lambda_{r}=1.9\%$), 0.25 ($\lambda_{m}=0.4\%$,
$\lambda_{r}=1.6\%$), and 1 ($\lambda_{m}=\lambda_{r}=1\%$).

\begin{figure}[h]
\begin{tabular}{cc}
\includegraphics[scale=0.21]{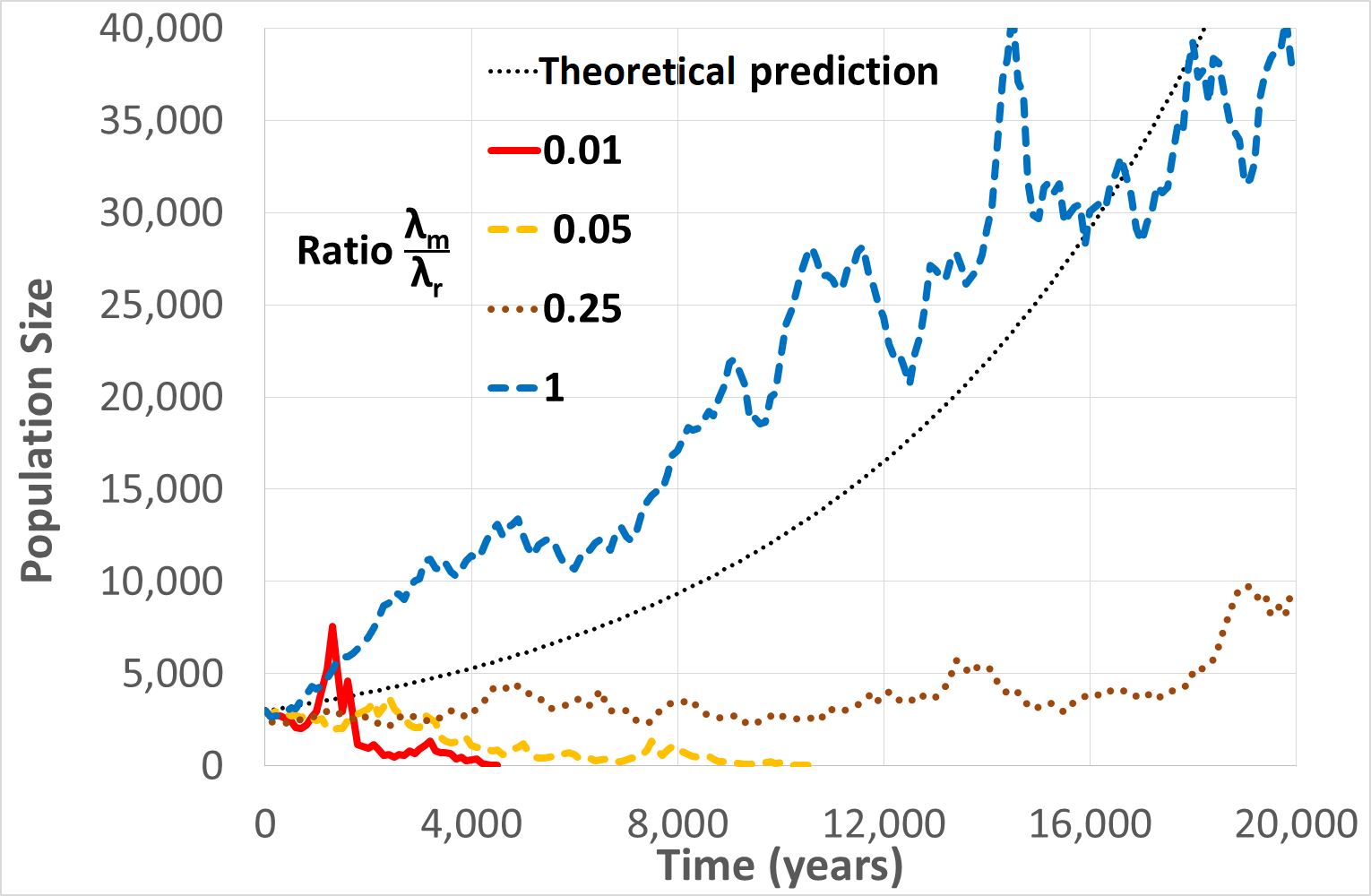}  & \includegraphics[scale=0.33]{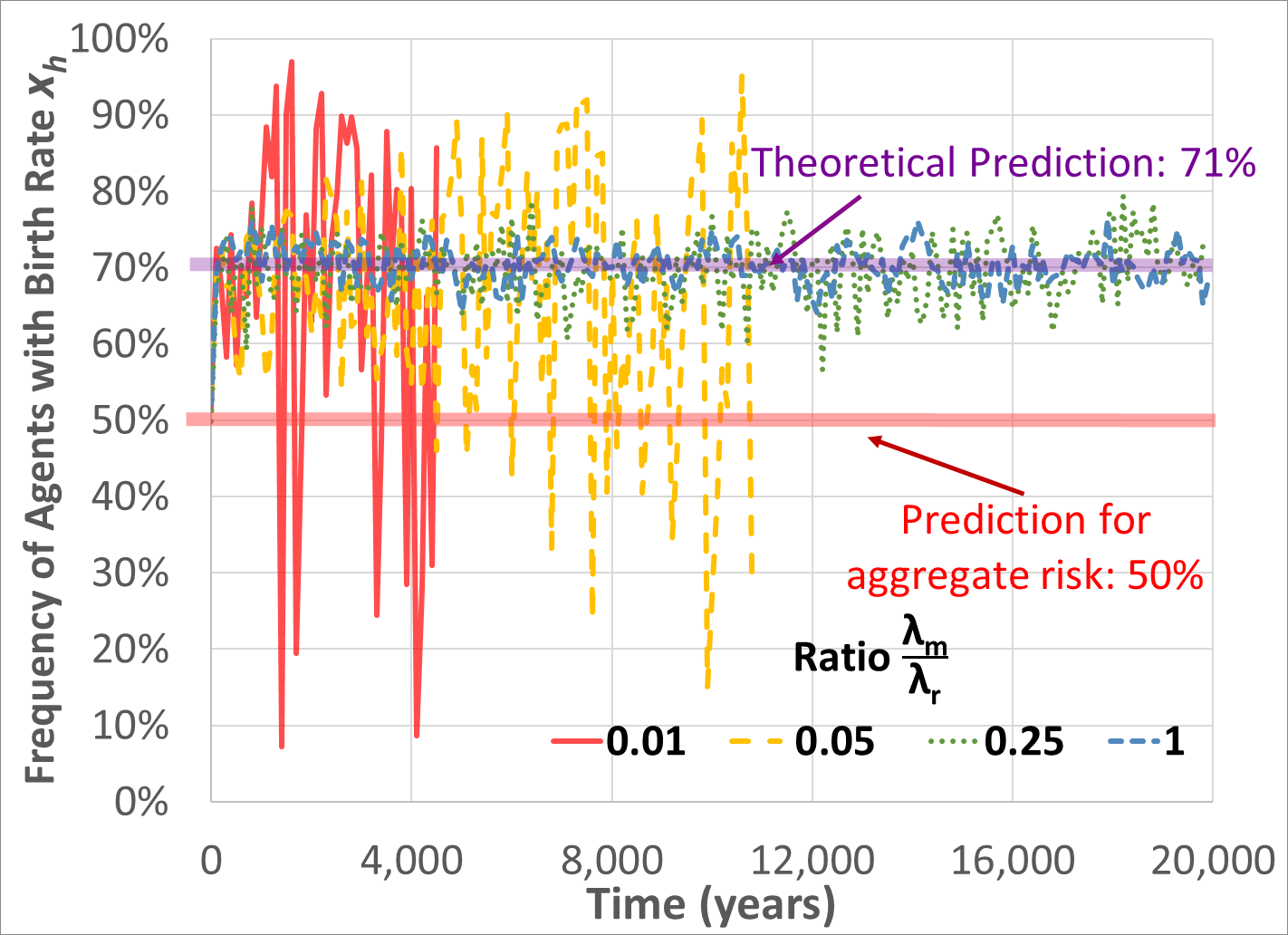}\tabularnewline
\includegraphics[scale=0.34]{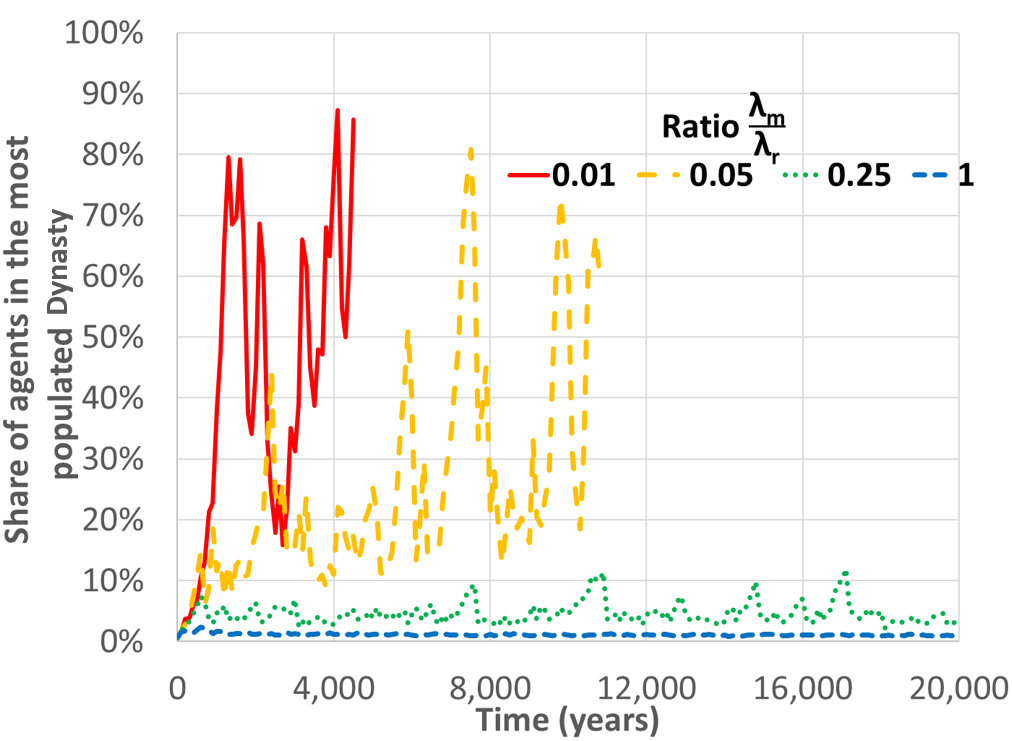}  & \includegraphics[scale=0.33]{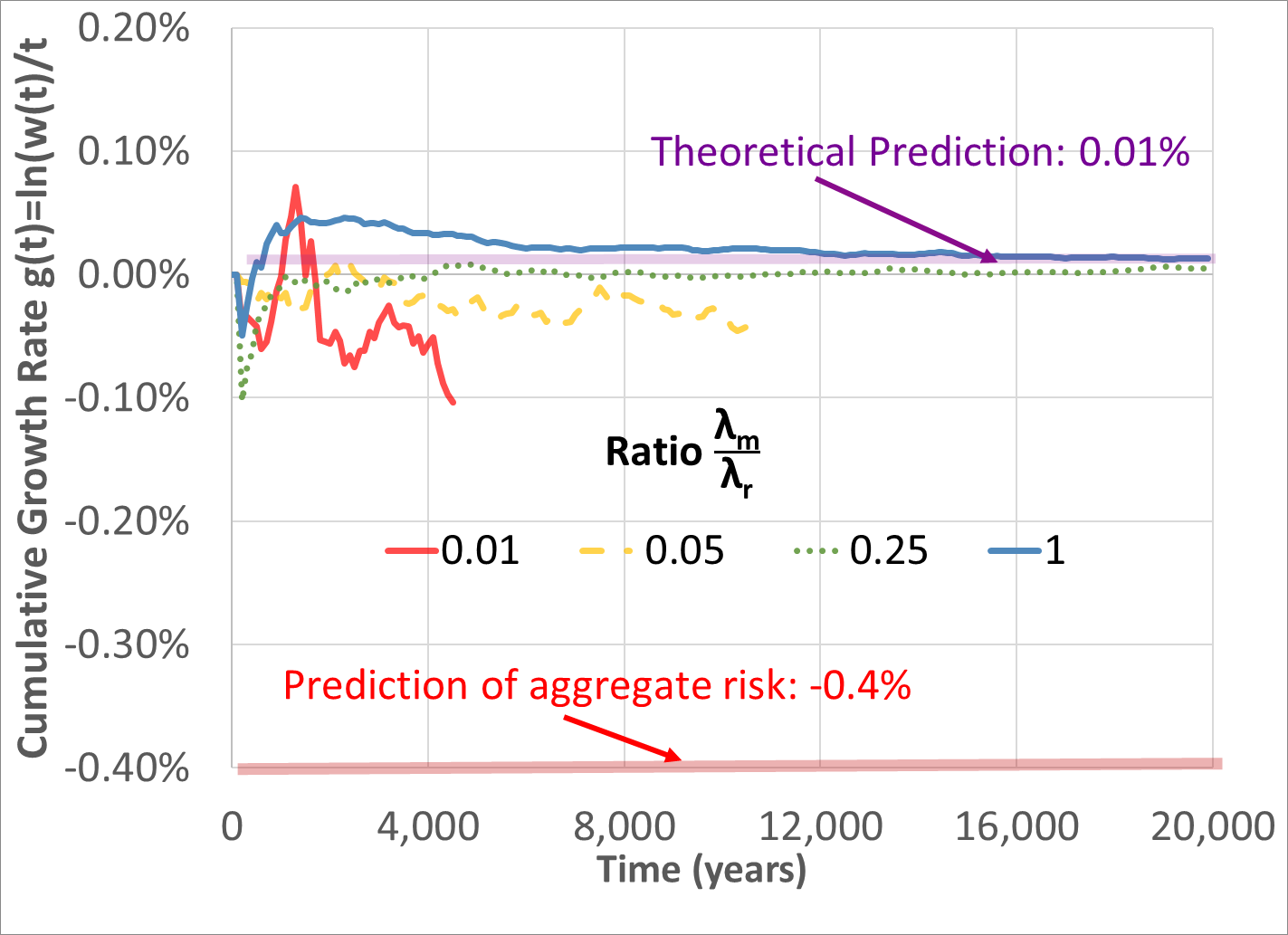}\tabularnewline
\end{tabular}~

\caption{\label{fig:Sample-Results-for}Representative Simulation Runs for
four Ratios of $\frac{\lambda_{m}}{\lambda_{r}}$}
\end{figure}

The top-left panel of Figure \ref{fig:Sample-Results-for} shows the
dynamics of the total population in each of the four simulation runs.
The top-right panel shows how the frequency of agents that are endowed
with a high {heritable} birth rate evolves. The bottom-left panel shows
the percentage of agents that live in the most populated {dynasty}
(among the 300 dynasties). The bottom-right panel shows the cumulative
growth rate up to time $t$ in each year (i.e., it shows $g\left(t\right)=\frac{\ln\left(w\left(t\right)\right)}{t}$).

The figure shows that when the ratio $\frac{\lambda_{m}}{\lambda_{r}}$
is small (0.01 or 0.05), dynastic risk has similar properties to aggregate
risk. The low rate of 
migration implies that a couple of ``successful''
dynasties (which happen to have had a high {heritable} birth rate for a
long time) contain most of the population. This causes the {heritable} risk,
essentially, to be aggregate. The frequency of agents with a high
{heritable} birth rate has large fluctuations, since a single change of
the {heritable} birth rate of the most populated {dynasty} has a large impact
on this frequency. This is shown in the top-right panel. The cumulative
growth rate (bottom-right panel) is initially positive, but after
a couple of thousand years it becomes negative and starts converging
to the negative growth predicted by aggregate risk, until the population
becomes extinct (top-left panel).

By contrast, Figure \ref{fig:Sample-Results-for} shows that when the ratio $\frac{\lambda_{m}}{\lambda_{r}}$
is 0.25 (resp., 1), then the theoretical prediction for the continuum case becomes relatively
(resp., very) accurate for the finite population. When the migration
rate is sufficiently high, a ``successful'' {dynasty} spreads its
offspring to many other dynasties, staving off extinction. The bottom-left
panel shows that the frequency of agents living in the most populated
{dynasty} is at most 10\% (resp., 2\%). This implies that the share
of agents with a high {heritable} birth rate has a relatively (resp., very)
small fluctuations around Claim \ref{cla:binary}'s predicted value
of about 71\%, as can be seen in the top-right panel. The cumulative
growth rate (bottom-right panel) converges to the positive value of
0.01\%, as predicted in Claim \ref{cla:binary}, as is shown in the
top-left panel.

\begin{figure}
\begin{centering}
\includegraphics[scale=0.44]{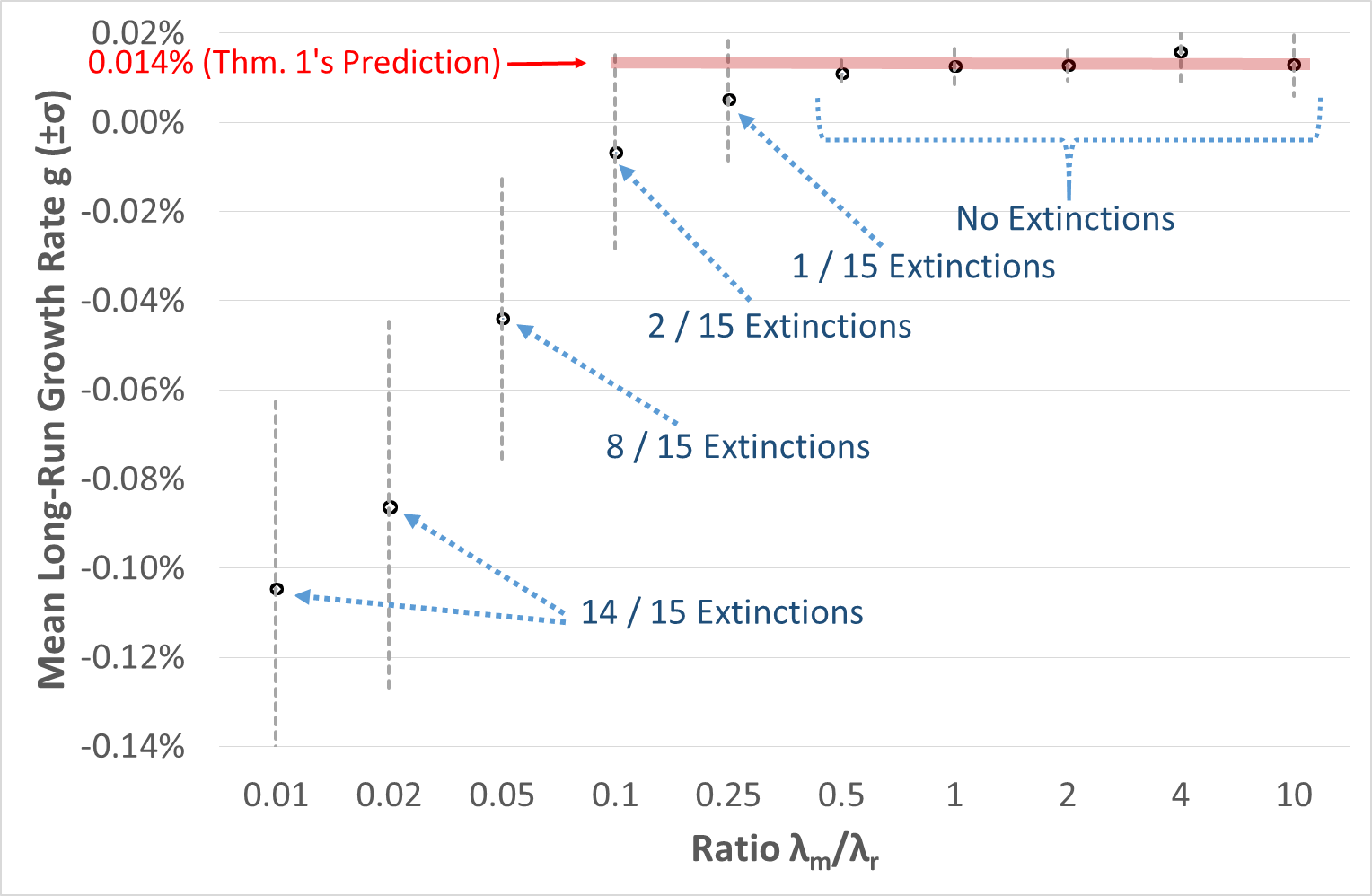}
\par\end{centering}
\caption{\label{fig2:Mean-Growth-Rate}Mean Long-Run Growth Rate for each Ratio
of $\frac{\lambda_{m}}{\lambda_{r}}$}
The black points describe the mean growth rate of 15 simulation runs
for each ratio of $\frac{\lambda_{m}}{\lambda_{r}}$. The vertical
bars show intervals of one standard deviation on each side of the
mean. The labels describe how many simulation runs ended in an extinction
of the population.
\end{figure} 
Figure \ref{fig2:Mean-Growth-Rate} presents the mean long-run growth
rate obtained in the 15 simulation runs for each of the ten ratios of $\frac{\lambda_{m}}{\lambda_{r}}$.
{The results show that conclusions drawn from the four representative simulation runs presented in Figure \ref{fig:Sample-Results-for} are indeed valid for the entire set of 150 simulation runs.}

\section{Discussion}\label{sec:Discussion}

\paragraph{Asexual reproduction}

	Our model, like the related literature, makes the simplifying assumption that reproduction is asexual, where offspring are identical to the parent. Similar results should hold if reproduction were sexual and haploid, where a single genetic variant---an allele---that determines choice is inherited with probability $1/2$ from either parent. That is, if a particular choice in a gamble is currently favored, this advantage will hold in a muted form if offspring inherit it through haploid sex. Further, if the gene controlling choice is evident to a mate, homophily---a preference for like individuals---would accentuate this advantage, bringing the model back to the asexual case.

\paragraph{Horizontal and vertical correlation}
An insight of our model is that vertical correlation increases the growth rate, but horizontal correlation decreases it. Horizontal correlation is called \emph{within-generation bet hedging }by \citeauthor{lehmann2007natural} (\citeyear{lehmann2007natural}). Vertical correlation is called the \emph{multiplayer effect} by \citeauthor{mcnamara2011evolution} (\citeyear{mcnamara2011evolution}) who study a non-overlapping generations model in which an asexual species breeds annually in one of a large number of breeding sites. Each site can be either good (high expected number of offspring) or bad. In each generation each site changes its type with probability less than 0.5. Each animal observes a noisy signal about the quality of the site in which it was born, and it has to choose whether to stay or to migrate to a new site. \citeauthor{mcnamara2011evolution} show that when the signal is sufficiently noisy, it is best for nature to induce each animal to ignore the signal, and always stay in its birth site because the mere fact that the animal was born in the site makes it more likely that the site is good. 

\paragraph{Additive separability}
Our model assumes that the various component of risk are additively separable. This assumption  clearly facilitates the analysis. It permits a direct comparison of the implications of the three types of risk. Separability seems intuitively unlikely to be crucial
to the results. At the least, there ought to be approximate results for a general non-separable criterion and small aggregate, heritable and idiosyncratic components. Further, it seems that it would be possible to allow for arbitrary aggregate shocks with heritable and idiosyncratic shocks conditional on the aggregate state, much as in \cite{robson1996biological}.
\paragraph{Age structure}

Recently, a different approach was applied by \cite{Robson-Samuelson-2019}
to show that risk interdependence matters in a continuous-time setting
(see also related results in \citealp{robson2009evolution}). Specifically,
they show that adding age structure to  \citeauthor{robatto2017biological}'s
(\citeyear{robatto2017biological}) setting (i.e., allowing the fertility
rate to depend on the agent's age) implies that interdependence of
risk influences the growth rate. By contrast, the present paper shows
that interdependence of risk is important for the induced growth rate
in a hierarchical population, even when the age structure is trivial, but still in a continuous-time setting. It would be interesting for
future research to study the implications of 
heritable risk in age-structured
populations.

\paragraph{Migration between fragmented habitats}

Our numerical analysis suggests an important advantage to connecting
isolated small habitats of an endangered species. The related existing
literature (e.g., \citealp{burkey1999extinction,smith2002population})
shows that having several isolated small habitats for a species induces
a larger extinction probability relative to a situation in which the
species lives in a single large habitat. This result holds in a setup
in which the birth rates are decreasing in the population's density,
and are deterministic. The present paper shows that connecting isolated
small habitats with migration increases the long-run growth rate.
We adopt a complementary setup of the birth rate that does not depend
on the population's density, but does have a dynastic stochastic component {(the heritable component of the birth rate}).

\section{Conclusion}\label{sec:Conclusions}

In this paper, we demonstrate that a crucial aspect of the evolution of a population exposed to risk is inheritance. If the actual choice made by a parent is inherited by her offspring, this 
induces a correlation between the parent's risk and the offspring's risk. A type that does this will outperform types that are exposed to either idiosyncratic or aggregate risk.
This result is a force favoring risk-taking. Although most risk-taking may be reversed by a sufficiently concave relationship between resources and offspring, positively skewed lotteries that involve high enough prizes, but relatively low means, will be taken.

\appendix

\section{\label{subsec:Proof-of-Theorem-one}Proof of Theorem \ref{thm:main-result}}

The following global convergence result of \cite{goh1978global} will
be helpful in the proof
\begin{lem}[{{{{{{{\citealp[Theorems 1 and 2]{goh1978global}}}}}}}}]
\label{lem:convergence-lemma}Consider the system of $n$ differential
equations
\[
\frac{dp_{k}\left(t\right)}{dt}=p_{k}\left(t\right)\cdot F_{k}\left(p\right),k=1,..,n,
\]
where each $F_{k}\left(p\right)$ is a continuous function of $p\in\boldsymbol{R}_{+}^{n}\equiv\left\{ p|p_{k}>0\,\forall k\in\left\{ 1,..,n\right\} \right\} $. Suppose there is a fixed point $p^{*}>{0}$ satisfying $0=p_k^{*}\cdot F_{k}\left(p^{*}\right)$
for each $k$. 
Assume further that there exists a constant matrix $E$ such that for all
$p\in\boldsymbol{R}_{+}^{n}$: (1) $\frac{\partial F_{k}\left(p\right)}{\partial p_{k}}\leq E_{kk}<0$
for each $k\in\left\{ 1,..,n\right\} $, and (2) $\left|\frac{\partial F_{k}\left(p\right)}{\partial p_{j}}\right|\leq E_{jk}$
for each $j\neq k$, and all the leading principal minors of $-E$
are positive. Then every trajectory $p\left(t\right)$ starting at
any initial state $p\left(0\right)>{0}$ converges to $p^{*}>{0}.$\footnote{\citeyear{goh1978global} before Theorem 1 and Theorem 2.
\citeauthor{goh1978global}'s (\citeyear{goh1978global}) Theorem 2 implicitly assumes the existence of a fixed point explicitly assumed in \citeauthor{goh1978global}'s (\citeyear{goh1978global}) Theorem 1. It follows that the fixed point is unique.}
\end{lem}
For each time $t$, let $w_{k}\left(t\right)$ be the mass of agents
with heritable birth rate $x_{k}$ at time $t$ (henceforth, $x_{k}$-agents).
Let $p_{k}\left(t\right)=\frac{w_{k}\left(t\right)}{w\left(t\right)}$
be the share of $x_{k}$-agents at time $t$\emph{.}{} Let $\bar{\boldsymbol{b}}\left(t\right)=\sum_{k}p_{k}\left(t\right)\cdot x_{k}+\mu_{y}+\boldsymbol{z}\left(t\right)$ be the average birth rate at time $t$.
Let $\boldsymbol{b}_{k}\left(t\right)$ be the average birth rate
of $x_{k}$-agents in time $t$: $\boldsymbol{b}_{k}\left(t\right)=x_{k}+\mu_{y}+\boldsymbol{z}\left(t\right)$.
The mass of $x_{k}$-agents at time $t+dt$ is given by (neglecting
terms of $O\left(\left(dt\right)^{2}\right)$):
\[
w_{k}\left(t+dt\right)=w_{k}\left(t\right)+dt\cdot\left(\left(\boldsymbol{b}_{k}\left(t\right)-\delta-\lambda\right)\cdot w_{k}\left(t\right)+w\left(t\right)\cdot\lambda_{x}\cdot q_{k}\right)\text{,}
\] 
Hence $$\frac{dw_k(t)}{dt}=\left(\boldsymbol{b}_{k}\left(t\right)-\delta-\lambda_x\right)\cdot w_{k}\left(t\right)+w\left(t\right)\cdot\lambda_{x}\cdot q_{k}\text{.}
$$

The mass of agents at time $t+dt$ is given by
\[
w\left(t+dt\right)=w\left(t\right)+dt\cdot\left(\bar{\boldsymbol{b}}\left(t\right)-\delta\right)\cdot w\left(t\right),
\]

so that $$\frac{dw}{dt}=\left(\bar{\boldsymbol{b}}\left(t\right)-\delta \right)\cdot w\left(t\right).$$

Since 
\[\frac{1}{p_k(t)}\frac{dp_k(t)}{dt}=\frac{1}{w_k(t)}\frac{dw_k(t)}{dt} -\frac{1}{w(t)}\frac{dw(t)}{dt},\]
it follows that 

$$
\frac{dp_{k}(t)}{dt}=\left(\boldsymbol{b}_{k}\left(t\right)-\bar{\boldsymbol{b}}\left(t\right)-\lambda_{x}\right)\cdot p_{k}\left(t\right)+\lambda_{x}\cdot q_{k}.
$$
Substituting $\bar{{x}}\left(t\right)\equiv\sum_{k}p_{k}\left(t\right)\cdot x_{k}$,
we obtain
\begin{equation}
\frac{dp_{k}\left(t\right)}{dt}=\left(\left(x_{k}-\bar{{x}}\left(t\right)\right)-\lambda_{x}\right)\cdot p_{k}\left(t\right)+\lambda_{x}\cdot q_{k}=p_{k}\left(t\right)\cdot F_{k}\left(p\right).\label{eq:pk-dt0}
\end{equation}

Let $\Delta_{+}^{n}\subseteq\boldsymbol{R}_{+}^{n}$ be the interior of the simplex,
$\Delta_{+}^{n}=\left\{ p\in\boldsymbol{R}_{+}^{n}|\sum_{k}p_{k}=1\right\} $.
Since $\frac{d}{dt}\sum_1^np_k(t)=0,$ it follows that $p\left(0\right)\in\Delta_{+}^{n}$ implies that $p\left(t\right)\in\Delta_{+}^{n}$
for each $t$.

We now show that there exists a fixed point $p^{*}\in\Delta_{+}^{n}$. If $\frac{dp_{k}\left(t\right)}{dt}=0$
and $p_{k}\left(t\right)=p_{k}^{*}$ in Eq. (\ref{eq:pk-dt0}), then setting $x^{*}\equiv\sum_{k}p_{k}^{*}\cdot x_{k}$ yields the requirement:
\begin{equation}
p_{k}^{*}=\frac{\lambda_{x}\cdot q_{k}}{\lambda_{x}+x^{*}-x_{k}},\label{eq:p*}
\end{equation}
where it will be shown that the denominator is positive, for all $k=1,...,n$. Next multiply each $p_k^{*}$ in (\ref{eq:p*}) by $x_{k}$ and
sum to obtain an equation in one unknown:

\begin{equation}
x^{*}=\sum_{k}\frac{x_{k}\cdot\lambda_{x}\cdot q_{k}}{\lambda_{x}+x^{*}-x_{k}}.\label{eq:x*}
\end{equation}
In the range $x^{*}>x_{n}-\lambda_{x}$
the LHS (resp., RHS) is increasing (resp., decreasing) in $x^{*}$. Further, LHS<RHS if $x^{*}-(x_{n}-\lambda_{x})>0$ but small enough and LHS>RHS if $x^{*}$ is large enough. These observations imply that there
exists a unique solution $x^{*}>x_{n}-\lambda_{x}$ to Eq. (\ref{eq:x*}).
This implies, using Eq (\ref{eq:p*}), that $p_{k}^{*}>0,k=1,...,n$ and that $\sum_{k=1}^np_{k}^{*}=1$ so that $p^{*}\in\Delta_{+}^{n}$.

We now prove {global} asymptotic convergence to this $p^{*}$ from any initial
state. We have
\begin{equation}
\frac{dp_{k}\left(t\right)}{dt}=p_{k}\left(t\right)\cdot F_{k}\left({p}\right),\,\,\,\,\,\,\,\,\,\,\,\,\,\,\textrm{where}\,\,\,\,F_{k}\left({p}\right)=\left(x_{k}-\bar{{x}}\left(t\right)\right)-\lambda_{x}+\frac{\lambda\cdot q_{k}}{p_{k}}.\label{eq:pk-dt}
\end{equation}
Taking the partial derivative of $F_{k}\left({p}\right)$ we obtain, for $j\neq k$:
\[
\left|\frac{\partial F_{k}\left(p\right)}{\partial p_{j}}\right|=-x_{j}<0,\,\,\,\,\,\,\textrm{and}
\]
\[
\frac{\partial F_{k}\left(p\right)}{\partial p_{k}}=-\left(x_{k}+\lambda_x\frac{q_{k}}{\left(p_{k}\right)^{2}}\right)<-x_{k}<-x_1<0.
\]

Let the matrix $E$ be equal to $-x_{1}$ on the main diagonal, and
equal to zero otherwise. Then all the conditions of Lemma \ref{lem:convergence-lemma}
are satisfied, which implies that $p\left(t\right)$ converges to $p^{*}\in\Delta_{+}^{n}$ from any initial state $p\left(0\right)\in\Delta_{+}^{n}$.
The fact that $p\left(t\right)$ converges to $p^{*}$ implies that
\[
\lim_{t\rightarrow\infty}\left|\bar{\boldsymbol{b}}\left(t\right)-x^{*}-\mu_{y}-\boldsymbol{z}\left(t\right)\right|=0.
\]

This, in turn, implies that the equivalent growth rate is given by:
\[
g=\lim_{t\rightarrow\infty}\frac{\log{w}\left(t\right)}{t}=f\left(X,q,\lambda_{x}\right)+\mu_{y}+\mu_{z}-\delta,
\]
\[
\textrm{where\,\,\,\,\,\,\,\,\,\,\,\,\,\,\,\,\,\,\,\,\,\,\,\,\,\,\,\,\,\,\,\,\,\,\,\,\,\,\,\,\,\,\,\,}\,\,\,\,\,\,f\left(X,q,\lambda_{x}\right)\equiv x^{*}\in\left(\max\left(\mu_{x},x_{n}-\lambda_{x}\right),x_{n}\right).\,\,\,\,\,\,\,\,\,\,\,\,\,\,\,\,\,\,\,\,\,\,\,\,\,\,\,\,\,\,\,\,\,\,\,\,\,\,\,\,\,\,\,\,\,\,\,\,\,\,\,\,\,\,\,\,\,\,\,\,\,\,\,\,\,\,\,\,\,\,
\]

We prove the final claim as follows. Let $g\left(x_{k},x^{*}\right)$
be defined as:
\[
g\left(x_{k},x^{*}\right)\equiv\frac{x_{k}\cdot\lambda_x}{\lambda_x+x^{*}-x_{k}}-x^{*}.
\]
Observe that $g\left(x_{k},x^{*}\right)$ is a strictly decreasing
function of $x^{*}$ (in the domain $x^{*}>x_{n}-\lambda_x).$ Next
we show that $g\left(x_{k},x^{*}\right)$ is strictly convex in $x_{k}$:

\[
\frac{\partial g\left(x_{k},x^{*}\right)}{\partial x_{k}}=\frac{\lambda_x\cdot\left(\lambda_x+x^{*}\right)}{\left(\lambda_x+x^{*}-x_{k}\right)^{2}}\,\Rightarrow\,\frac{\partial^{2}g\left(x_{k},x^{*}\right)}{\partial\left(x_{k}\right)^{2}}=\frac{2\cdot\lambda_x\cdot\left(\lambda_x+x^{*}\right)}{\left(\lambda_x+x^{*}-x_{k}\right)^{3}}>0.
\]
\[
.
\]

Eq. (\ref{eq:x*}) is equivalent to $\boldsymbol{E}_{\boldsymbol{x}}\left[g\left(\boldsymbol{x},x^{*}\right)\right]=0$.
The convexity of $g\left(x_{k},x^{*}\right)$ implies that $\boldsymbol{E}_{\boldsymbol{x}}\left[g\left(\boldsymbol{x},x^{*}\right)\right]$
increases following a mean preserving spread from $\boldsymbol{x}=\left(X,q_{x}\right)$
to $\boldsymbol{x}'=\left(X',q_{x'}\right)$, which, in turn, implies
that the unique solution $x^{*}$ to $\boldsymbol{E}_{\boldsymbol{x}}\left[g\left(\boldsymbol{x},x^{*}\right)\right]=0$
must strictly increase as well, since $g\left(x_{k},x^{*}\right)$
is strictly decreasing in $x^{*}$. In particular, this implies that
$x^{*}>\mu_{x}.$

\section{\label{subsec:Proof-of-Lemma}Proof of Proposition \ref{smallbeta}} 
The growth rate derived from the mean $\bar{c}=\sum_{c\in C}q(c)c$ is $\bar{c}^\beta$. The growth rate derived from the lottery $(C,q)$ is the unique 
solution for $x^{*}>m^{\beta}-\lambda_x$ of \begin{equation} \label{xstar}
x^{*}=\lambda_{x}\sum_{c\in C}\frac{q(c)\cdot c^\beta}{\lambda_{x}+x^{*}-c^\beta}.
\end{equation}

If $\beta>0$ is sufficiently small, it follows that $x^{*}<\bar{c}^{\beta}$ if \begin{equation*}
R\equiv\bar{c}^{\beta}>\lambda_{x}\sum_{c\in C}\frac{q(c) \cdot c^{\beta}}{\lambda_{x}+\bar{c}^{\beta}-c^\beta}\equiv S.
\end{equation*} This is because the LHS of Eq. (\ref{xstar}) is increasing in $x^*$ and the RHS of Eq. (\ref{xstar}) is decreasing in $x^*$ 
for $x^*>m^\beta - \lambda_x.$ In addition, $\bar{c}^\beta>m^\beta - \lambda_x$, for sufficiently small $\beta>0$. 

We have $R=S=1$ at $\beta =0$. In addition,
\begin{equation}
\label{eq:dlhs}
\left.\frac{dS}{d\beta}\right|_{\beta=0}=\left(1+\frac{1}{\lambda_{x}}\right)\sum_{c\in C}q(c)\ln c-\frac{1}{\lambda_{x}}\ln\bar{c}<\ln(\bar{c})=\left.\frac{dR}{d\beta}\right|_{\beta=0} 
\end{equation}
Hence $R>S$ for all small enough $\beta >0$
due to the concavity of the function $\ln(x)$.

\section{Explicit Solution for Binary Lotteries\label{sec:Explicit-Solution-for}}

Theorem \ref{thm:main-result} has derived the key properties of the
growth rate, without calculating an explicit formula for $f\left(X,q_{x},\lambda_{x}\right)$.
In what follows we present such an explicit formula in the case of
binary lotteries over the heritable birth rate, which is used to yield
the theoretical predictions in Section \ref{sec:Numeric-Analysis-of}
and in Figure \ref{fig:illus}.
Specifically, we now assume that the heritable birth rate has two possible
realizations, i.e., $X=\left\{ x_{l},x_{h}\right\} $. Let $\mu_x$
denote the lottery's\emph{ }expectation, let $\Delta x=x_{h}-x_{l}$
denotes the lottery's \emph{spread, }and let $q\equiv q_x\left(x_{h}\right)$
denote the probability of the higher realization.
\begin{claim}
\label{cla:binary}The equivalent growth rate of a growth process
with a binary heritable birth rate is equal to $g=f\left(\Delta x,\mu_x,q,\lambda_{x}\right)+\mu_{y}+\mu_{z}-\delta,$ where
\begin{equation}
f\left(\Delta x,\mu_{x},q,\lambda_{x}\right)=\mu_{x}+\frac{\Delta x\cdot\left(1-2\cdot q\right)-\lambda_{x}+\sqrt{\left(\Delta x-\lambda_{x}\right)^{2}+4\cdot q\cdot\Delta x\cdot\lambda_{x}}}{2}.\label{eq:f-explicit}
\end{equation}
Moreover, $f\left(\Delta x,\mu_{x},q,\lambda_{x}\right)$ is decreasing
in $\lambda_{x}$.
\end{claim}
\begin{proof}[Proof]
Substituting $p_{h}=p$, $p_{l}=1-p$, $q_{h}=q$ and $q_{l}=1-q$
in Eq. (\ref{eq:p-main}) yields:

\[
p=\frac{\lambda_x\cdot q}{\lambda_x-(1-p)\cdot\Delta x}\,\,\Leftrightarrow\,\,p^{2}\cdot\Delta x+p\cdot\left(\lambda_{x}-\Delta x\right)-q\cdot\lambda_{x}=0.
\]
This quadratic equation has a unique solution in $\left(0,1\right)$:
\begin{equation}
p(\Delta x,q,\lambda_{x})=\frac{\Delta x-\lambda_{x}+\sqrt{\left(\Delta x-\lambda_{x}\right)^{2}+4\cdot q\cdot\Delta x\cdot\lambda_{x}}}{2\cdot\Delta x},\label{eq:p*-binary}
\end{equation}
which yields (\ref{eq:f-explicit}), when substituting this solution into
\[
f\left(\Delta x,\mu_{x},q,\lambda_{x}\right)=p(\Delta x,q,\lambda_{x})\cdot x_{h}+\left(1-p(\Delta x,q,\lambda_{x})\right)\cdot x_{\ell}=\mu_{x}+\left(p(\Delta x,q,\lambda_{x})-q\right)\cdot\Delta x.
\]
Next we prove that $f\left(\Delta x,\mu_{x},q,\lambda_{x}\right)$
is decreasing in $\lambda_{x}$. Take the derivative of $p(\Delta x,q,\lambda_{x})$:
\[
\frac{\partial p(\Delta x,q,\lambda_{x})}{\partial\lambda_{x}}=\frac{1}{2\cdot\Delta x}\left(\frac{-2\cdot\left(\Delta x-\lambda_{x}\right)+4\cdot q\cdot\Delta x}{2\cdot\sqrt{\left(\Delta x-\lambda_{x}\right)^{2}+4\cdot q\cdot\Delta x\cdot\lambda_{x}}}-1\right).
\]
We have to show that $\frac{\partial p(\Delta x,q,\lambda_{x})}{\partial\lambda_{x}}$
is negative for any $\lambda_{x}>0$, which is true iff
\[
\sqrt{\left(\Delta x-\lambda_{x}\right)^{2}+4\cdot q\cdot\Delta x\cdot\lambda_{x}}\overset{}{>}\Delta x\cdot\left(2\cdot q-1\right)+\lambda_{x}
\]
After some algebra, this condition holds if and only if 
$q(1-q)>0$ which is true for all $q\in\left(0,1\right)$.
\end{proof}
\bibliographystyle{chicago}
\bibliography{risk}

\end{document}